\documentclass[conference]{IEEEtran}

\usepackage{amssymb}
\usepackage{amsmath}
\usepackage{graphicx}
\usepackage{cite}
\usepackage{subfigure}
\usepackage{graphicx,epstopdf}
\usepackage{epsfig}	
\usepackage{cite,graphicx,amsmath,amssymb,amsthm}
\usepackage{comment}
\usepackage{lipsum}

\usepackage{bm}
\usepackage{textcomp}

\usepackage{algorithm}
\usepackage[noend]{algpseudocode}
\usepackage{indentfirst}

\makeatletter
\def\BState{\State\hskip-\ALG@thistlm}
\makeatother

\usepackage{amssymb}
\usepackage{amsmath}
\usepackage{graphicx}
\usepackage{cite}
\usepackage{citesort}
\usepackage{balance}
\usepackage[utf8x]{inputenc}

\newtheorem{cor}{Corollary}

\IEEEoverridecommandlockouts

\usepackage{graphicx,epstopdf}
\usepackage{epsfig}	
\usepackage{amsfonts,balance}
\usepackage{bbm}
\floatname{algorithm}{Algorithm}
\usepackage{lipsum}

\newtheorem{theorem}{Theorem}

\newtheorem{corollary}{Corollary}

\makeatletter
\def\ScaleIfNeeded{%
\ifdim\Gin@nat@width>\linewidth \linewidth \else \Gin@nat@width
\fi } \makeatother

\begin{document}

\title{Molecular Communication with a Reversible Adsorption Receiver}
\author{
\IEEEauthorblockN{ Yansha Deng\IEEEauthorrefmark{1}, Adam~Noel\IEEEauthorrefmark{2}, Maged~Elkashlan\IEEEauthorrefmark{3}, 
Arumugam~Nallanathan\IEEEauthorrefmark{1}, 
and Karen C. Cheung\IEEEauthorrefmark{2}  
 } \IEEEauthorblockA{
\IEEEauthorrefmark{1}Department of Informatics, King's College London, London, UK\\
\IEEEauthorrefmark{2}Department of Electrical and Computer Engineering,
University of British Columbia, Vancouver, BC, Canada \\
\IEEEauthorrefmark{3}School of Electronic Engineering and Computer
Science, Queen Mary University of London, London\\
 } }

\maketitle

\begin{abstract}
In this paper, we present an analytical model for  a diffusive molecular communication (MC) system with a reversible adsorption receiver in a fluid environment.  The time-varying  spatial distribution of the  information molecules  under  the reversible adsorption and desorption reaction at the surface of a bio-receiver is analytically characterized.  Based on the spatial distribution, we derive the  number of  newly-adsorbed information molecules expected in any time duration.  Importantly, we present a simulation framework for the proposed model that accounts for the diffusion and reversible reaction. Simulation results show the accuracy of our derived expressions, and demonstrate the positive effect of the adsorption rate and the negative effect of the desorption rate on the net number of  newly-adsorbed information molecules expected. Moreover, our analytical results simplify to the special case of an absorbing receiver.
\end{abstract}

\maketitle

\section{Introduction}

 
Conveying information over a distance has been a problem  over decades, and is  urgently demanded for different dimensions and various environments. The conventional solution is to utilize  electrical- or electromagnetic-enabled  communication, which  is unfortunately inapplicable or inappropriate in very small dimensions or in specific  environments, such as in
salt water, tunnels, or human bodies.
Recent breakthroughs in bio-nano technology have motivated  molecular communication \cite{nakano2013molecu} to be  a biologically-inspired technique for nanonetworks, where devices with functional components on the scale of 1 to 100 nanometers, namely nanomachines, share information over distance via chemical signals  in nanometer to micrometer scale environments. 



Diffusion-based MC is the most simple,  general and energy efficient transportation  paradigm without the need for external energy or infrastructure,  where molecules propagate via the random motion, namely Brownian motion,  caused by collisions  
 with the fluid's molecules. Examples  include  deoxyribonucleic acid (DNA) signaling among DNA segments \cite{howard1993radom630} and  calcium signaling among cells \cite{nakano2005molecular}. 

In a practical bio-inspired system, the surface of a receiver  is covered with  selective receptors, which  are sensitive to a specific type of information molecule (e.g., specific peptides or calcium ions). The surface of the receiver  may adsorb or  bind with this specific information molecule \cite{rospars2000perireceptor}. One example is  that the influx of calcium towards the  center of a receiver (e.g., cell) is induced by  the reception of a calcium signal \cite{bush2010nanoscale}.
 Despite  growing research efforts, the chemical reaction receiver is rarely accurately modeled and characterized in most of the  literature except  the works from Yilmaz \cite{Yilmaz2014effect,yilmaz2014three844,yilmaz2014simulation} and Chou \cite{chou2013extended}, since the local reactions  complicate the solution of the reaction-diffusion equations. 


Unlike   existing works on MC, we consider the reversible adsorption and desorption ({A$\&$D}) receiver, which is capable of adsorbing a certain type of information molecule near its surface, and desorbing the information molecules previously adsorbed at its surface. {A$\&$D}  is a widely-observed  process for   colloids \cite{FEDER1980144}, proteins \cite{ramsden1993concentration}, and polymers \cite{fang2005kinetics}. Also, the A$\&$D  process simplifies to the special case of an infinitely absorbing receiver. 
However, its modeling, analysis, and simulation  in the MC domain have never been investigated since the dynamic concentration change near the surface is more challenging than existing works with a passive receiver or an absorbing receiver. 

 From a theoretical perspective, researchers  have  derived the equilibrium concentration  of A$\&$D \cite{andrews2009accurate}, which is insufficient to model the time-varying channel impulse response (and ultimately the communications performance) of an {A$\&$D} receiver. 
From a simulation perspective,  the simulation design for the  A$\&$D process of molecules at the surface of a \emph{planar} receiver   was proposed   in \cite{andrews2009accurate}. However, the simulation procedure for the  \emph{MC communication system} with  a \emph{spherical} A$\&$D  receiver, where the information  molecules,  triggered by the transmission of multiple pulses,   propagate via free-diffusion through the channel, and contribute to the received signal through  A$\&$D at the surface of the receiver, has never been solved and reported. 
This is  due to the complexity in modeling the coupling effect of adsorption and desorption under diffusion, as well as accurately and dynamically tracking the location and the number of diffused molecules, adsorbed molecules and desorbed molecules.


Despite the aforementioned challenges, in this paper we  consider  the  diffusion-based MC system with a point transmitter and an  {A$\&$D} receiver.  The  goal of this paper is to characterize the impact of  the {A$\&$D}  receiver on the net number of newly-adsorbed molecules expected. Our major contributions are summarized
as follows.
\begin{enumerate}
\item We present an analytical model for the diffusion-based MC system with an {A$\&$D} receiver. We derive the exact expression for the channel impulse response at a spherical {A$\&$D} receiver  in a three dimensional (3D) fluid environment  due to a single release of multiple molecules (single transmission).
  We then derive the \emph{net} number of \emph{newly}-adsorbed molecules expected at the surface of the {A$\&$D} receiver in any time duration. 
\item We propose a simulation algorithm to simulate the diffusion, adsorption and desorption behavior of information molecules  based on a particle-based simulation framework.  Unlike  existing simulation platforms (e.g., Smoldyn \cite{andrews2010detailed}, NanoNS \cite{gul2010nanons}), our simulation algorithm captures the  dynamic process of  a MC system, which are the molecule   emission,  free diffusion, and  \text{A$\&$D} at the surface of the receiver.   Our simulation results are in close agreement with the derived number of  adsorbed molecules expected. 
\end{enumerate}

The rest of this paper is organized as follows. In Section II, we introduce the system model  . In Section III, we present the channel impulse response of information molecules.  In Section IV, we  present the simulation framework. In Section V, we discuss the numerical and simulation results. In Section VI, we conclude our contributions.

\section{System Model}
We consider a 3-dimensional (3D) diffusion-based MC system in a fluid environment with a point  transmitter and  a spherical A$\&$D  receiver. We assume  spherical symmetry where the transmitter is \emph{effectively} a spherical shell and the molecules are released from random points over the shell;  the actual angle to the transmitter when a molecule hits the receiver is ignored, so this assumption cannot accommodate a flowing environment.
 The point transmitter  is located at a distance $r_0$ from the center of the receiver and is at a distance $d=r_0-r_r$ from the nearest point on the surface of the receiver with radius $r_r$.   The extension to an asymmetric spherical model that accounts for the actual angle to the transmitter when a molecule hits the receiver complicates the derivation of the channel impulse response, and may be solved following \cite{scheider1972two}.  
 
 We assume all receptors are equivalent and can accommodate at most one adsorbed molecule.
 The ability of a molecule to adsorb at a given site is independent of the occupation of neighboring receptors.
The spherical receiver is assumed to have  no physical limitation on the number of molecules adsorbed to the receiver surface (i.e., we ignore saturation). This is an appropriate assumption for  a sufficiently low number of adsorbed molecules, or  for a sufficiently high concentration of receptors.
  We also assume  perfect synchronization between the transmitter and the receiver as in most literature \cite{Yilmaz2014effect,yilmaz2014simulation,yilmaz2014three844}. We consider three processes: emission, propagation, and reception,  which are detailed in the following.

\subsection{Emission}

The point transmitter releases one type of information molecule (e.g., hormones, pheromones, or deoxyribonucleic acid (DNA)) to the receiver for information transmission. The transmitter emits $N_{\rm{tx}}$ information  molecules at $t=0$, where we define the initial condition as \cite[3.61]{schulten2000lectures}
 \begin{align}
C\left( {r,\left. {t \to 0} \right|{r_0}} \right) = \frac{1}{{4\pi {r_0}^2}}\delta \left( {r - {r_0}} \right), \label{initial1}
\end{align}
where $C\left( {r,\left. {t \to 0} \right|{r_0}} \right)$  is the molecule distribution function at time ${t \to 0}$ and distance $r$ with initial distance $r_0$.

We also define the first boundary condition as
 \begin{align}
\mathop {\lim }\limits_{r \to \infty } C\left( {r,\left. t \right|{r_0}} \right) = 0, \label{boundary3}
\end{align}
such that  a molecule that diffuses extremely far away from the receiver is effectively removed from the fluid environment.

\subsection{Diffusion}

 Once the information molecules are emitted, they diffuse by randomly colliding with other molecules in the environment. This random motion is called Brownian motion \cite{howard1993radom630}.
The concentration of information molecules is assumed to be sufficiently low that the collisions between those information molecules are ignored \cite{howard1993radom630}, such that each information molecule diffuses independently with constant diffusion coefficient  $D$.
 The propagation model in a 3D environment is described by  Fick's second law \cite{howard1993radom630,yilmaz2014three844}:
\begin{align}
\frac{{\partial \left( {r \cdot C\left( {r,\left. t \right|{r_0}} \right)} \right)}}{{\partial t}} = D\frac{{{\partial ^2}\left( {r \cdot C\left( {r,\left. t \right|{r_0}} \right)} \right)}}{{\partial {r^2}}}, \label{ficklaw}
\end{align}
where  the diffusion coefficient 
is found experimentally  \cite{philip2008biological587}.

\subsection{Reception}
We consider the reversible A$\&$D receiver, which is capable of counting the net number of newly-adsorbed molecules at the surface of the receiver.  Any molecule that  hits the receiver surface is either adsorbed to the receiver surface  or reflected back into the fluid environment, based on the adsorption rate $k_1$ (length$\times $time$^{-1}$). The adsorbed molecules   either desorb  or remain stationary at the surface of receiver, based on the desorption rate $k_{-1}$ (time$^{-1}$). 

At $t=0$, there are no information molecules at the receiver surface, so the second initial condition is 
\begin{align}
C\left( {{r_r},\left. 0 \right|{r_0}} \right)=0, \text{and}\; {C_a}\left( {\left. 0 \right|{r_0}} \right) = 0, \label{initial2}
\end{align}
where ${C_a}\left( {\left. t \right|{r_0}} \right)$ is the average concentration of  molecules that are adsorbed to the receiver surface at time $t$.

 For the solid-fluid interface located at $r_r$, the second boundary condition of the information molecules is  \cite{andrews2009accurate}
\begin{align}
{\left. {D\frac{{\partial \left( {C\left( {r,\left. t \right|{r_0}} \right)} \right)}}{{\partial r}}} \right|_{r = r_r^ + }} = {k_1}C\left( {{r_r},\left. t \right|{r_0}} \right) - {k_{ - 1}}{C_a}\left( {\left. t \right|{r_0}} \right), \label{boundary1}
\end{align}	
where  $k_1$ and $k_{-1}$ are  non-zero finite constants.  Here, the adsorption rate $k_1$ is approximately limited to the thermal velocity of potential adsorbents (e.g., $k_1 < 7 \times 10^6$ $\rm{\mu m /s}$ for a 50 kDa protein at 37 $^\circ$C) \cite{andrews2009accurate}; the desorption rate $k_{-1}$  is typically from $10^{-4}$ $s^{-1} $ and  $10^{4}$ $s^{-1}$ \cite{Tom2007Multistage}.

The surface concentration ${C_a}\left( {\left. t \right|{r_0}} \right)$ changes over time as follows:
\begin{align}
\frac{{\partial {C_a}\left( {\left. t \right|{r_0}} \right)}}{{\partial t}} = {\left. {D\frac{{\partial \left( {C\left( {r,\left. t \right|{r_0}} \right)} \right)}}{{\partial r}}} \right|_{r = r_r^ + }}, \label{boundary2}
\end{align}
which shows that the change  in the adsorbed concentration over time is equal to the flux of diffusion molecules towards the surface. 

Combining \eqref{boundary1} and  \eqref{boundary2}, we write
\begin{align}
\frac{{\partial {C_a}\left( {\left. t \right|{r_0}} \right)}}{{\partial t}} = {k_1}C\left( {{r_r},\left. t \right|{r_0}} \right) - {k_{ - 1}}{C_a}\left( {\left. t \right|{r_0}} \right), \label{boundary3}
\end{align}
which is known as the Robin or radiation boundary condition, and shows that the equivalent adsorption rate is proportional to the molecule concentration at the surface.

\section{Receiver Observations }
In this section, we first derive the spherically-symmetric spatial distribution  ${C\left( {r,\left. t \right|{r_0}} \right)}
$, which is the probability of finding a molecule at distance $r$ and time $t$. We then derive the  flux at  the surface of the A$\&$D receiver, from which we  derive the exact   number of adsorbed molecules expected at the surface of the receiver.
 In the following theorem, we solve the time-varying spatial distribution of information molecules at the surface of the receiver.

\begin{theorem}
The expected time-varying spatial distribution of an information molecule released into a 3D fluid environment with a reversible adsorbing receiver is given by
\begin{align}
C\left( {r,\left. t \right|{r_0}} \right) =  & \;  \frac{1}{{4\pi {r_0}r\sqrt {4\pi Dt} }}\exp \left\{ { - \frac{{{{\left( {r - {r_0}} \right)}^2}}}{{4Dt}}} \right\}
\nonumber\\&  + \frac{1}{{4\pi {r_0}r\sqrt {4\pi Dt} }}\exp \left\{ { - \frac{{{{\left( {r + {r_0} - 2{r_r}} \right)}^2}}}{{4Dt}}} \right\}
\nonumber\\&  - \frac{1}{{2\pi r}}\int_0^\infty  {\left( {{e^{ - jwt}}\varphi _Z^*\left( w \right) + {e^{jwt}}{\varphi _Z}\left( w \right)} \right)dw},
\label{der20}
\end{align}
where 
 \begin{align}
{\varphi _Z}\left( w \right) = & \; Z\left( {jw} \right) = \frac{{2\left( {\frac{1}{{{r_r}}} + \frac{{{k_1}jw}}{{D\left( {jw + {k_{-1}}} \right)}}} \right)}}{{\left( {\frac{1}{{{r_r}}} + \frac{{{k_1}jw}}{{D\left( {jw + {k_{-1}}} \right)}} + \sqrt {\frac{{jw}}{D}} } \right)}}
\nonumber\\& \times \frac{1}{{4\pi {r_0}\sqrt {4Djw} }}\exp \left\{ { - \left( {r + {r_0} - 2{r_r}} \right)\sqrt {\frac{{jw}}{D}} } \right\}.
\label{der19}
\end{align}
and $\varphi _Z^*\left( w \right)$ is the complex conjugate of $\varphi _Z\left( w \right)$.
\end{theorem}
\begin{proof}
See Appendix A.
\end{proof}

We observe that \eqref{der20} reduces to the \emph{absorbing} receiver \cite[Eq. (3.99)]{schulten2000lectures} when there is no desorption (i.e., $k_{-1}=0$).

To characterize the number of information molecules adsorbed at the surface of the receiver using ${ C\left( {\left. {r,t} \right|{r_0}} \right)}$,
we define the  rate of the coupled reaction (i.e.,  adsorption and desorption)  at the surface of the reversible adsorbing receiver   as \cite[Eq. (3.106)]{schulten2000lectures}
\begin{align}
K\left( {\left. t \right|{r_0}} \right) = 4\pi r_r^2D{\left. {\frac{{\partial C\left( {\left. {r,t} \right|{r_0}} \right)}}{{\partial r}}} \right|_{r = {r_r}}}.
\label{der21}
\end{align}

\begin{cor}
The rate of the coupling reaction at  the surface of a reversible adsorbing receiver is given by
\begin{align}
K\left( {\left. t \right|{r_0}} \right) = & \; 2{r_r}D{\int_{ 0 }^\infty  {{e^{ - jwt}}\left[ {\sqrt {\frac{{jw}}{D}} {\varphi _Z}\left( w \right)} \right]} ^*}dw
\nonumber\\& + 2{r_r}D{\int_{ 0 }^\infty  {{e^{  jwt}}\left[ {\sqrt {\frac{{jw}}{D}} {\varphi _Z}\left( w \right)} \right]} }dw,
\label{der22}
\end{align}
where ${\varphi _Z}\left( w \right)$ is  as given in \eqref{der19}. 
\end{cor}
\begin{proof}
By substituting \eqref{der20} into \eqref{der21}, we derive the  coupling reaction rate at the surface of an A$\&$D receiver as \eqref{der22}.
\end{proof}

From \textbf{Corollary 1}, we can derive the  net change in the number of adsorbed molecules expected for any time interval in the following theorem.

\begin{theorem}
The  net change in the number of  adsorbed molecules expected  at the surface of the receiver during the interval [$T$, $T$+$T_{s} $]  is derived as 
\begin{align}
& \mathbb{E}\left[ {{N_{\rm{A\&D}}}\left( {\left. {{\Omega _{{r_r}}},T,T + T_{s}} \right|{r_0}} \right)} \right] = 2{r_r}{N_{\rm{tx}}}D \nonumber\\
  &     \hspace{0.3cm} \times \Bigg[ {\int_{ 0 }^\infty  \frac{{{e^{ - jwT}} - {e^{ - jw\left( {T + T_{s}} \right)}}}}{{jw}}\Big[ {\sqrt {\frac{{jw}}{D}} {\varphi _Z}\left( w \right)} \Big] ^*}dw \Bigg.
  \nonumber\\
  &     \hspace{0.3cm} \Bigg.+  {\int_{ 0 }^\infty  {\frac{{{e^{  jw\left( {T + T_{s}} \right)}}-{e^{  jwT}} }}{{jw}}\Big[ {\sqrt {\frac{{jw}}{D}} {\varphi _Z}\left( w \right)} \Big]} }dw \Bigg]
  ,
\label{der241}
\end{align}
where ${\varphi _Z}\left( w \right)$ is given in \eqref{der19},  $T_{s}$ is the sampling time, and  ${\Omega _{{r_r}}}$ represents the spherical receiver with radius $r_r$.
\end{theorem}

\begin{proof}
The cumulative fraction of particles that are adsorbed at the surface of the  receiver at time $T$  is expressed as
\begin{align}
{{R_{\rm{A\&D}}}\left( {\left. {{\Omega _{{r_r}}},T} \right|{r_0}} \right)}&  = \int_0^T {K\left( {\left. t \right|{r_0}} \right)dt}.
\label{der23}
\end{align}

Based on \eqref{der23}, the  net change of adsorbed molecules expected at the surface of the receiver during the interval [$T$, $T$+$T_{s} $] is defined as
\begin{align}
\mathbb{E}&\left[ {{N_{{\rm{A\&D}}}}\left( {\left. {{\Omega _{{r_r}}},T,T + T_{s}} \right|{r_0}} \right)} \right] = 
\nonumber\\& {N_{\rm{tx}}}{R_{\rm{A\&D}}}\left( {\left. {{\Omega _{{r_r}}},T + T_{s}} \right|{r_0}} \right)   - {N_{\rm{tx}}}{R_{\rm{A\&D}}}\left( {\left. {{\Omega _{{r_r}}},T} \right|{r_0}} \right).
\label{der24}
\end{align}

Substituting \eqref{der23} into \eqref{der24}, we derive the expected net change of adsorbed molecules during any observation interval as \eqref{der241}.
\end{proof}

\section{Simulation Framework}
This section describes the stochastic simulation framework of the point-to-point MC system with the A$\&$D receiver described by \eqref{boundary1}. 
To   accurately capture the locations of individual information molecules, we adopt a particle-based simulation framework with a spatial resolution on the  scale of several nanometers \cite{andrews2009accurate}.

\subsection{Algorithm}
We  present the algorithm for simulating the MC system with an  A$\&$D receiver in Algorithm 1. 
In the following subsections, we describe the details of Algorithm 1.
\begin{algorithm}
\caption{Simulation of a MC System with an A$\&$D Receiver}\label{euclid}
Require: $N_{\rm{tx}}$, $r_0$, $r_r$,  $ {\Omega _{{r_r}}}$, $D$, $\Delta t$, $T_s$, $T_b$, $N_{\rm{th}}$
\begin{algorithmic}[1]
\Procedure{Initialization}{}
\State Determine Simulation End Time
\State  Add $N_{\text{tx}}$ emitted molecules 
\BState  \textbf{For all} Simulation Time Step \textbf{do} 
\BState \,  \textbf{For all} free molecules in environment \textbf{do} 
\State  Propagate  free molecules following  $\mathcal{N}\left( {0,2D\Delta t} \right)$
\State  Evaluate  distance $d_m$ of  molecule to receiver  
   \If {$d_m< r_r$}
\State     Update  state $\&$ location of  collided molecule 
\State   Update $\#$ of collided  molecules  $N_C$
\EndIf
\BState \,   \textbf{For all} $N_C$ collided molecules  \textbf{do}
 \If {Adsorption Occurs}
\State  Update $\#$ of \emph{newly}-adsorbed molecules  $N_A$
\State  Calculate adsorbed molecule  location 
\State \, \, $\left( {x_m^A,y_m^A,z_m^A} \right)$
\Else  
\State Reflect the molecule off receiver surface to
\State \, \, $\left( {x_m^{Bo},y_m^{Bo},z_m^{Bo}} \right)$
\EndIf
\BState  \, \textbf{For all} \emph{previously}-adsorbed molecules \textbf{do} 
\If {Desorption Occurs}
\State    Update  state $\&$ location of  desorbed molecule 
\State  Update $\#$  of \emph{newly}-desorbed molecules  $N_D$
\State  Displace \emph{newly}-desorbed molecule to 
\State $\left( {x_m^D,y_m^D,z_m^D} \right)$
\EndIf
\BState \, Calculate \emph{net}  number of \emph{newly}-adsorbed molecules, 
\BState \, which is   $N_A-N_D$ 
\EndProcedure
\end{algorithmic}
\end{algorithm}

\subsection{ Emission and Diffusion}
 At  time $t=0$,   $N_{\rm{tx}}$ molecules are emitted from the point transmitter at a distance $r_0$ from the center of the receiver. 
The time is divided into small simulation intervals of size $\Delta t$, and each time instant is represented by $t_m=m\Delta t$, where $m$ is the current simulation index. 
The displacement  $\Delta S$ of a molecule in a 3D fluid environment in one simulation step $\Delta t$ is modeled as 
\begin{align}
 \Delta S = \left\{ \mathcal{N}\left({0,2D\Delta t}\right), \; \mathcal{N}\left({0,2D\Delta t}\right), \;\mathcal{N}\left({0,2D\Delta t}\right) \right\},
 \end{align}
where  ${\mathcal{N}\left( {0,2D\Delta t} \right)}$ is the normal distribution.
In each simulation step, the number of molecules and their locations are stored.
\subsection{Adsorption or Reflection}
 According to the second boundary condition in \eqref{boundary2},  molecules that collide with the receiver surface are either adsorbed or reflected back. The $N_C$ collided  molecules    are identified  by  calculating the distance between each molecule and the center of the receiver. Among the collided molecules, the probability of a molecule being adsorbed to the receiver surface, i.e., the adsorption probability, is a function of the diffusion coefficient, which is given as
 \cite[Eq. (10)]{erban2007reactive}
\begin{align}
{P_A} = {k_1}\sqrt {\frac{{\pi \Delta t}}{D}} . \label{abpro}
 \end{align}
 
The probability  that a collided molecule bounces off of the receiver is $1-{P_A}$.

 It is known that  adsorption may occur during the simulation step $\Delta t$, and  determining exactly where a molecule  adsorbed to the surface of the receiver during $\Delta t$ is a non-trivial problem. To simplify this, we assume that the adsorbed location of a  molecule during $[t_{m-1}, 
 t_m]$ is equal to the location where the line, formed by this molecule's location at the start of the current simulation step $\left( {{x_{m - 1}},{y_{m - 1}},{z_{m - 1}}} \right)$ and this molecule's location  at the end of the current simulation step after diffusion $\left( {{x_{m}},{y_{m}},{z_{m}}} \right)$, intersects the surface of the receiver. Assuming that the location of the center of receiver is $(x_r,y_r,z_r)$, then the location of the intersection point between this 3D line segment,
 and a sphere with center at  $(x_r,y_r,z_r)$ in the $m$th simulation step, can be shown to be
 \begin{align}
x_m^A = & {x_{m - 1}} + \frac{{{x_m} - {x_{m - 1}}}}{{\Delta }}g,\label{loc1}
\\ y_m^A = & {y_{m - 1}} + \frac{{{y_m} - {y_{m - 1}}}}{{\Delta }}g, \label{loc11}
\end{align}
\begin{align}
z_m^A = & {z_{m - 1}} + \frac{{{z_m} - {z_{m - 1}}}}{{\Delta }}g, \label{loc12}
 \end{align}
where 
\begin{align}
\Delta = \sqrt {{{\left( {{x_m} - {x_{m - 1}}} \right)}^2} + {{\left( {{y_m} - {y_{m - 1}}} \right)}^2} + {{\left( {{z_m} - {z_{m - 1}}} \right)}^2}} , \label{loc_delta}
\end{align}
\begin{align}
g = \frac{{ - b - \sqrt {{b^2} - 4ac} }}{{2a}}. \label{loc2}
\end{align}

In \eqref{loc2}, we have 
\begin{align}
a = &{\left( {\frac{{{x_m} - {x_{m - 1}}}}{{\Delta }}} \right)^2} + {\left( {\frac{{{y_m} - {y_{m - 1}}}}{{\Delta }}} \right)^2} + {\left( {\frac{{{z_m} - {z_{m - 1}}}}{{\Delta }}} \right)^2},\nonumber\\
b = &2\frac{{\left( {{x_m} - {x_{m - 1}}} \right) {({x_{m - 1}}-x_r)} }}{{\Delta }}  + 2\frac{{\left( {{y_m} - {y_{m - 1}}} \right)}{{(y_{m - 1}-y_r)} }}{{\Delta }}
\nonumber\\&    + 2\frac{{\left( {{z_m} - {z_{m - 1}}} \right) {({z_{m - 1}}-z_r) } }}{{\Delta }},\label{loc31}
\end{align}
\begin{align}
c = & { {({x_{m - 1}}-x_r)}^2} + { {({y_{m - 1}}-y_r)}^2} + {{({z_{m - 1}}-z_r)}^2}-{r_r}^2,
 \label{loc3}
\end{align}
where $\Delta$ is given in \eqref{loc_delta}.

 Of course, due to symmetry, the location of the adsorption site does not impact the overall accuracy of the  simulation.

If a molecule fails to adsorb to the receiver, then in the reflection process we make the approximation that the molecule bounces back to its position at the start of the current  simulation step.
 Thus, the location of the molecule after reflection by the receiver in the $m$th simulation step is approximated as
\begin{align}
\left( {x_m^{Bo},y_m^{Bo},z_m^{Bo}} \right) = \left( {{x_{m - 1}},{y_{m - 1}},{z_{m - 1}}} \right). \label{loc4}
\end{align}  

Note that the approximations for molecule locations in the adsorption process and the reflection process can be  accurate for  sufficiently small simulation steps (e.g., $\Delta t < 10^{-7}$  s for the system that we simulate in Section V), but   small simulation steps result in poor computational efficiency.  

\subsection{Desorption}
In the desorption process, the molecules adsorbed at the receiver boundary either desorb or remain adsorbed.  The desorption process can be modeled as a  first-order chemical reaction. 
 Thus,  the desorption probability of a molecule at the receiver surface  during  $\Delta t$ is given by  \cite[Eq. (22)]{andrews2009accurate}
\begin{align}
{P_D} = 1 - {e^{ - {k_{ - 1}}\Delta t}}.\label{loc5}
\end{align}

\begin{figure}[t!]
    \begin{center}
        \includegraphics[width=3.0 in,height=2.3in]{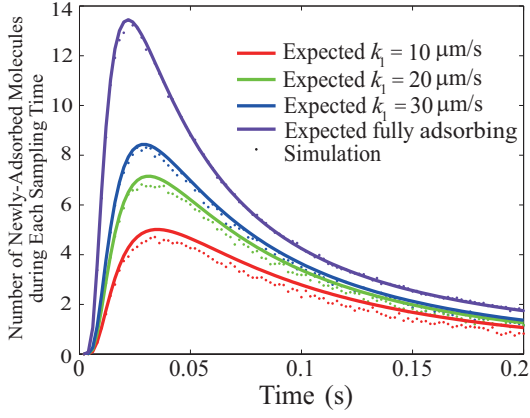}
        \vspace{-0.3cm}
        \caption{The net number of newly-adsorbed  molecules for various adsorption rates with parameters:
        $k_{-1} = 5$ $s^{-1} $, $N_{\rm{tx}} = 1000$, $r_r = 10$ $\rm{\mu m}$, $d = 1 \; \rm{\mu m} $, $D = 8$ $\rm{\mu m^2/s} $,   $T_s =0.002$ s, and the simulation step $\Delta t =10^{-5}$ s.}
        \label{fig:1}
    \end{center}
\end{figure}

The displacement of a  molecule after desorption is an important factor for accurate modeling of  molecule behaviour.  If the simulation step were   small, then we might place the desorbed molecule near the receiver surface; otherwise, doing so may result in an artificially higher chance of re-adsorption in the following time step, resulting in an inexact  concentration profile. To avoid this, we take into account the diffusion \emph{after} desorption, and place the desorbed molecule away from the surface with displacement $\left( {\Delta x,\Delta y,\Delta z} \right)$ 
\begin{align}
\left( {\Delta x,\Delta y,\Delta z} \right) = \left( {f\left( {{P_1}} \right),f\left( {{P_2}} \right),f\left( {{P_3}} \right)} \right), \label{loc6}
\end{align} 
where  each component was empirically found to be \cite[Eq. (27)]{andrews2009accurate}
\begin{align}
f\left( P \right) = \sqrt {2D\Delta t} \frac{{0.571825P - 0.552246{P^2}}}{{1 - 1.53908P + 0.546424{P^2}}}. \label{loc7}
\end{align} 

In \eqref{loc6}, $P_1$, $P_2$ and $P_3$ are  uniform random numbers between 0 and 1. Placing the desorbed molecule at a random distance away along the line from the center of the receiver to where the molecule was adsorbed  is not sufficiently accurate due to the lack of consideration for the coupling effect of  A$\&$D  and the diffusion coefficient  in \eqref{loc7}. 

Different from \cite{andrews2009accurate}, we have the spherical receiver  such that the   molecule
after desorption in our model should be  displaced differently.
We assume that the location of a molecule after desorption $\left( {x_m^D,y_m^D,z_m^D} \right)$, based on its location at the start of the current simulation step and the location of the center of the receiver $(x_r,y_r,z_r)$, can be approximated as
\begin{align}
x_m^D = & {x_{m - 1}^A} + {\rm{sgn}}\left( {{x_{m - 1}^A} - {x_r}} \right)\Delta x,
\nonumber\\ y_m^D = & {y_{m - 1}^A} + {\rm{sgn}}\left( {{y_{m - 1}^A} - {y_r}} \right)\Delta y,
\nonumber\\ z_m^D = & {z_{m - 1}^A} + {\rm{sgn}}\left( {{z_{m - 1}^A} - {z_r}} \right) \Delta z.
\label{loc8}
\end{align}
In \eqref{loc8}, $\Delta x$, $\Delta y$, and $\Delta z$ are given in \eqref{loc6}, and $\rm{sgn}\left(  \cdot  \right)$ is the Sign function.

\subsection{ Reception}
The receiver is capable of counting the net change in the number of adsorbed molecules in each simulation step.

\section{Numerical Results}

Fig. \ref{fig:1} and Fig. \ref{fig:2} plot the  net change of  adsorbed molecules   at the surface of the A$\&$D receiver at each sampling time $T_s$ due to a single bit transmission. The expected analytical curves are plotted using the  exact result in \eqref{der241}. The simulation points are plotted by measuring the net change of  adsorbed molecules  during $ [t,t + T_s]$   using  Algorithm 1 described in  Section IV, where $t = nT_s$, and $n \in \{1,2,3,\ldots \} $. In both figures, we  average  the   number of newly-adsorbed molecules expected over 1000 independent emissions of $N_{\rm{tx}}$ information molecules.  We see that the expected number of newly-adsorbed molecules measured using simulation is close to the exact analytical curves. Note that the small gap between the curves results from the local approximations in the adsorption, reflection, and desorption processes in \eqref{abpro}-\eqref{loc12}, \eqref{loc4}, and \eqref{loc8}, which can be reduced by setting smaller simulation step.

\begin{figure}[t!]
    \begin{center}
        \includegraphics[width=3.0 in,height=2.3in]{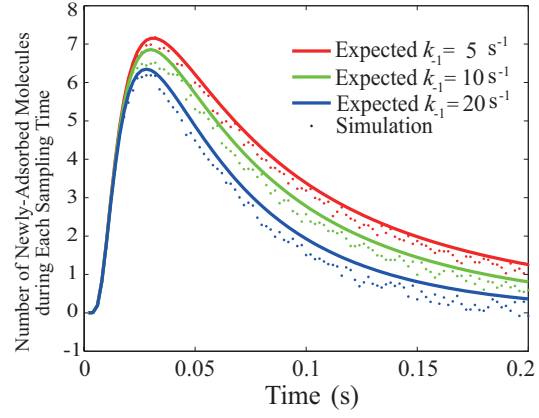}
        \vspace{-0.3cm}
        \caption{The net number of newly-adsorbed information molecules for various desorption rates with parameters:
        $k_1 = 20$ $\rm{\mu m /s}$, $N_{\rm{tx}} = 1000$, $r_r = 10 \; \rm{\mu m}$, $d = 1 \; \rm{\mu m} $, $D = 8$ $\rm{\mu m^2/s} $, $\Delta t =10^{-4}$ s,  and $T_s =0.002$ s. }
        \label{fig:2}
    \end{center}
\end{figure}

Fig. \ref{fig:1} examines the impact of the adsorption rate on the net number of newly-adsorbed  molecules expected at the surface of the receiver.  We fix the desorption rate to be $k_{-1} = 5\;$$\rm{s^{-1}}$. The number of newly-adsorbed  molecules  expected increases with increasing  adsorption rate $k_1$,  as predicted by \eqref{boundary1}.  Compared with the full adsorption receiver (e.i., $k_1$ $=\infty$), the A$\&$D receiver has a weaker observed signal. 
 \mbox{Fig. \ref{fig:2}}  shows the impact of the desorption rate on the  number of newly-adsorbed molecules expected at the surface of the receiver. We set $k_1 = 20$ $\mu \rm{m /s}$. The number of newly-adsorbed  molecules expected  decreases with increasing  desorption rate $k_{-1}$, which is as predicted by \eqref{boundary1}.
 
At the receiver side, the  number of newly-adsorbed molecules during each symbol interval could be compared with a threshold to demodulate the signal.   From the communication point of view,  
Fig. \ref{fig:1}   shows that 
 the   higher adsorption rate makes the received signal more distinguishable. 
In Fig.\ref{fig:1} and Fig. \ref{fig:2},  the shorter tail due to the lower adsorption rate and the higher desorption rate corresponds to less intersymbol interference.

\section{Conclusion}

In this paper, we modeled the diffusion-based MC system with the A$\&$D receiver. We derived the exact expression for the net number of newly-adsorbed information molecules expected at the surface of the receiver. We also presented a simulation algorithm that captures the behavior of each information molecule with the stochastic reversible reaction at the receiver. 
 We revealed that the number of  newly-adsorbed information molecules expected at the surface of the receiver increases with increasing  adsorption rate and with decreasing  desorption rate. Our ongoing work is comparing our proposed model with existing receiver models and considering the impact on bit error performance. Our analytical model and simulation framework provide a foundation for the accurate design and analysis of a more complex and realistic  receiver in molecular communication.
 
\appendices

\section{Proof of Theorem 1}\label{app_gsc_exact}

We first partition the spherically symmetric distribution into two parts using the method applied in \cite{schulten2000lectures}
\begin{align}
r{\cdot}C\left( {r,\left. t \right|{r_0}} \right) = r{\cdot}g\left( {r,\left. t \right|{r_0}} \right) 
+ r{\cdot}h\left( {r,\left. t \right|{r_0}} \right)
, \label{der1}
\end{align}
where
\begin{align}
g\left( {r,\left. {t \to 0} \right|{r_0}} \right) = \frac{1}{{4\pi {r_0}}}\delta \left( {r - {r_0}} \right)
, \label{der11}
\end{align}
\begin{align}
h\left( {r,\left. {t \to 0} \right|{r_0}} \right) = 0
. \label{der12}
\end{align}

Then, by substituting \eqref{der1} into \eqref{ficklaw},  we have
\begin{align}
\frac{{\partial \left( {r \cdot g\left( {r,\left. t \right|{r_0}} \right)} \right)}}{{\partial t}} = D\frac{{{\partial ^2}\left( {r \cdot g\left( {r,\left. t \right|{r_0}} \right)} \right)}}{{\partial {r^2}}}
, \label{der2}
\end{align}
and
\begin{align}
\frac{{\partial \left( {r \cdot h\left( {r,\left. t \right|{r_0}} \right)} \right)}}{{\partial t}} = D\frac{{{\partial ^2}\left( {r \cdot h\left( {r,\left. t \right|{r_0}} \right)} \right)}}{{\partial {r^2}}}
. \label{der3}
\end{align}

To derive $g\left( {r,\left. t \right|{r_0}} \right)$, we perform a Fourier transformation on $ r g\left( {r,\left. t \right|{r_0}} \right)$ to yield
\begin{align}
G\left( {k,\left. t \right|{r_0}} \right) = \int_{ - \infty }^\infty  {rg\left( {r,\left. t \right|{r_0}} \right){e^{ - ikr}}dr}
, \label{der41}
\end{align}
and
\begin{align}
r \cdot g\left( {r,\left. t \right|{r_0}} \right) = \frac{1}{{2\pi }}\int_{ - \infty }^\infty  {G\left( {k,\left. t \right|{r_0}} \right){e^{ikr}}dk} 
. \label{der42}
\end{align}

We then perform the Fourier transformation on \eqref{der2} to yield
\begin{align}
\frac{{dG\left( {k,\left. t \right|{r_0}} \right)}}{{dt}} =  - D{k^2}G\left( {k,\left. t \right|{r_0}} \right)
. \label{der5}
\end{align}

According to \eqref{der5} and the uniqueness of the Fourier transform, we derive 
\begin{align}
G\left( {k,\left. t \right|{r_0}} \right) = {K_g}\exp \left\{ { - D{k^2}t} \right\}
,\label{der51}
\end{align}
where $K_g$ is an undetermined constant.

The Fourier transformation  performed on \eqref{der11} yields
\begin{align}
G\left( {r,\left. {t \to 0} \right|{r_0}} \right) = \frac{1}{{4\pi {r_0}}}{e^{ - ik{r_0}}}
. \label{der111}
\end{align}

Combining \eqref{der51} and \eqref{der111}, we arrive at 
\begin{align}
G\left( {k,\left. t \right|{r_0}} \right) = \frac{1}{{4\pi {r_0}}}{e^{ - ik{r_0}}}\exp \left\{ { - D{k^2}t} \right\}
. \label{der6}
\end{align}

Substituting \eqref{der6} into \eqref{der42}, we find that
\begin{align}
r \cdot g\left( {r,\left. t \right|{r_0}} \right) = \frac{1}{{8\pi {r_0}\sqrt {\pi Dt} }}\exp \left\{ { - \frac{{{{\left( {r - {r_0}} \right)}^2}}}{{4Dt}}} \right\}
. \label{der52}
\end{align}

By performing the Laplace transform on  \eqref{der52}, we write
\begin{align}
\mathcal{L}\left\{ {r \cdot g\left( {r,\left. t \right|{r_0}} \right)} \right\} = \frac{1}{{4\pi {r_0}\sqrt {4Ds} }}\exp \left\{ { - \left| {r - {r_0}} \right|\sqrt {\frac{s}{D}} } \right\}
. \label{der53}
\end{align}

We then focus on solving the solution $h\left( {k,\left. t \right|{r_0}} \right)$ by first performing the Laplace transform on $h\left( {k,\left. t \right|{r_0}} \right)$ and \eqref{der3} as
\begin{align}
H\left( {r,\left. s \right|{r_0}} \right)= \mathcal{L}\left\{ {  h\left( {r,\left. t \right|{r_0}} \right)} \right\} = \int_0^\infty  {h\left( {r,\left. t \right|{r_0}} \right){e^{ - s\tau }}d\tau } 
,\label{der7}
\end{align}
and 
\begin{align}
srH\left( {r,\left. s \right|{r_0}} \right) = D\frac{{{\partial ^2}\left( {rH\left( {r,\left. s \right|{r_0}} \right)} \right)}}{{\partial {r^2}}}
,\label{der8}
\end{align}
respectively.

According to \eqref{der8}, the Laplace transform of the solution with respect to the boundary condition in \eqref{der8} is 
\begin{align}
rH\left( {r,\left. s \right|{r_0}} \right) = f\left( s \right)\exp \left\{ { - \sqrt {\frac{s}{D}} r} \right\}
,\label{der9}
\end{align}
where $f\left( s \right)$ needs to satisfy the second initial condition in \eqref{initial2}, and the second boundary condition in  \eqref{boundary1} and \eqref{boundary2}.

Having the Laplace transform of $\{r \cdot g\left( {r,\left. t \right|{r_0}} \right)\}$ and ${  h\left( {r,\left. t \right|{r_0}} \right)}$ in \eqref{der53} and \eqref{der9}, and performing a Laplace transformation on \eqref{der1}, we derive 
\begin{align}
&r\tilde C\left( {{r},\left. s \right|{r_0}} \right) = G\left( {r,\left. s \right|{r_0}} \right) + rH\left( {r,\left. s \right|{r_0}} \right)
\nonumber\\ &=   \frac{1}{{8\pi {r_0}\sqrt {Ds} }}\exp \left\{ { - \left| {r - {r_0}} \right|\sqrt {\frac{s}{D}} } \right\} 
+ f\left( s \right)\exp \left\{ { - \sqrt {\frac{s}{D}} r} \right\},
\label{der10}
\end{align}
where $\tilde C\left( {r,\left. s \right|{r_0}} \right) = \int_0^\infty  {C\left( {r,\left. t \right|{r_0}} \right){e^{ - st}}dt} $.

To solve $f\left( s \right)$, we perform the Laplace transform on the Robin boundary condition in \eqref{boundary3} to yield
\begin{align}
{{\tilde C}_a}\left( {\left. s \right|{r_0}} \right) = \frac{{{k_1}\tilde C\left( {{r_r},\left. s \right|{r_0}} \right)}}{{s + {k_{ - 1}}}}
,\label{der1111}
\end{align}
where  ${{\tilde C}_a}\left( {r,\left. s \right|{r_0}} \right) = \int_0^\infty  {{C_a}\left( {r,\left. t \right|{r_0}} \right){e^{ - st}}dt} $.

We then perform the Laplace transform on  the second initial condition in \eqref{initial2} and the second boundary condition in \eqref{boundary1} as
\begin{align}
{\left. {D\frac{{\partial \left( {\tilde C\left( {r,\left. t \right|{r_0}} \right)} \right)}}{{\partial r}}} \right|_{r = {r_r}}} = {k_1}\tilde C\left( {{r_r},\left. s \right|{r_0}} \right) - {k_{ - 1}}{{\tilde C}_a}\left( {\left. s \right|{r_0}} \right). \label{der121}
\end{align}

Substituting \eqref{der111} into \eqref{der121}, we obtain
\begin{align}
{\left. {D\frac{{\partial \left( {\tilde C\left( {r,\left. t \right|{r_0}} \right)} \right)}}{{\partial r}}} \right|_{r = {r_r}}}   =   \frac{{{k_1}s}}{{s + {k_{ - 1}}}}\tilde C\left( {{r_r},\left. s \right|{r_0}} \right)
.\label{der12}
\end{align}

To facilitate the analysis, we express the Laplace transform on the second boundary condition  as 
\begin{align}
{\left. {\frac{{\partial \left( {r \cdot \tilde C\left( {r,\left. s \right|{r_0}} \right)} \right)}}{{\partial r}}} \right|_{r = {r_r}}} = \left( {1 + \frac{{{r_r}{k_1}s}}{{D\left( {s + {k_{ - 1}}} \right)}}} \right)\tilde C\left( {r,\left. s \right|{r_0}} \right)
.\label{der13}
\end{align}

Substituting \eqref{der10} into \eqref{der13}, we determine $f\left( s \right)$ as
\begin{align}
&f\left( s \right) = \frac{{\left( {\sqrt {\frac{s}{D}}  - \frac{1}{{{r_r}}} - \frac{{{k_1}s}}{{D\left( {s + {k_{-1}}} \right)}}} \right)}}{{\left( {\sqrt {\frac{s}{D}}  + \frac{1}{{{r_r}}} + \frac{{{k_1}s}}{{D\left( {s + {k_{-1}}} \right)}}} \right)}}\frac{{\exp \left\{ { - \left( {{r_0} - 2{r_r}} \right)\sqrt {\frac{s}{D}} } \right\}}}{{4\pi {r_0}\sqrt {4Ds} }}.
\label{der15}
\end{align}

Having \eqref{der10} and \eqref{der15}, and performing the Laplace transform of the  concentration distribution,  we derive 
\begin{align}
&r \tilde C\left( {r,\left. s \right|{r_0}} \right) = \frac{1}{{4\pi {r_0}\sqrt {4Ds} }}\exp \left\{ { - \left| {r - {r_0}} \right|\sqrt {\frac{s}{D}} } \right\}
\nonumber\\& \hspace{0.5cm} + \frac{1}{{4\pi {r_0}\sqrt {4Ds} }}\exp \left\{ { - \left( {r + {r_0} - 2{r_r}} \right)\sqrt {\frac{s}{D}} } \right\}
\nonumber\\& \hspace{0.5cm} - \underbrace {  \frac{{2\left( {\frac{1}{{{r_r}}} + \frac{{{k_1}s}}{{D\left( {s + {k_{-1}}} \right)}}} \right)}}{{\left( {\frac{1}{{{r_r}}} + \frac{{{k_1}s}}{{D\left( {s + {k_{-1}}} \right)}} + \sqrt {\frac{s}{D}} } \right)}}\frac{{\exp \left\{ { - \left( {r + {r_0} - 2{r_r}} \right)\sqrt {\frac{s}{D}} } \right\}}}{{4\pi {r_0}\sqrt {4Ds} }}}_{Z\left( s \right)}.
\label{der16}
\end{align}

Applying the inverse Laplace transform  leads to 
\begin{align}
&r C\left( {r,\left. s \right|{r_0}} \right) = \frac{1}{{8\pi {r_0}\sqrt {\pi Dt} }}\exp \left\{ { - \frac{{{{\left( {r - {r_0}} \right)}^2}}}{{4Dt}}} \right\} +
\nonumber\\& \hspace{0.5cm} \frac{1}{{8\pi {r_0}\sqrt {\pi Dt} }}\exp \left\{ { - \frac{{{{\left( {r + {r_0} - 2{r_r}} \right)}^2}}}{{4Dt}}} \right\} -{\mathcal{L}^{ - 1}}\left\{ {Z\left( s \right)} \right\}.
\label{der17}
\end{align}

Due to the complexity of $Z(s)$, we can not derive the closed-form expression for its inverse Laplace transform ${f_z}\left( t \right) = {\mathcal{L}^{ - 1}}\left\{ {Z\left( s \right)} \right\}$. We  employ the Gil-Pelaez theorem \cite{wendel1961} for the characteristic function to derive  the cumulative distribution function (CDF) ${F_z}\left( t \right) $ as
\begin{align}
{F_z}\left( t \right) 
  =  & \frac{1}{2} - \frac{1}{\pi }\int_0^\infty  {\frac{{{e^{ - jwt}}\varphi _Z^*\left( w \right) - {e^{jwt}}{\varphi _Z}\left( w \right)}}{{2jw}}} dw,
\label{der181}
\end{align}
where 
${\varphi _Z}\left( w \right)$ is given in \eqref{der19}.

Taking the derivative of ${F_z}\left( t \right)$, we derive the inverse Laplace transform of $Z(s)$ as
\begin{align}
&{f_z}\left( t \right) = \frac{1}{{2\pi }}\int_0^\infty  {\left( {{e^{ - jwt}}\varphi _Z^*\left( w \right) + {e^{jwt}}{\varphi _Z}\left( w \right)} \right)dw} .
\label{der18}
\end{align}

Combining \eqref{der17} and \eqref{der19}, we finally derive the expected time-varying spatial distribution in \eqref{der20}.
%

\begin{thebibliography}{10}


\bibitem{nakano2013molecu}
N.~Tadashi, A.~W. Eckford, and T.~Haraguchi, \emph{Molecular Communication},
  1st~ed.\hskip 1em plus 0.5em minus 0.4em\relax Cambridge: Cambridge
  University Press, 2013.

\bibitem{howard1993radom630}
H.~C. Berg, \emph{Random Walks in Biology}.\hskip 1em plus 0.5em minus
  0.4em\relax Princeton University Press, 1993.

\bibitem{nakano2005molecular}
T.~Nakano, T.~Suda, M.~Moore, R.~Egashira, A.~Enomoto, and K.~Arima,
  ``Molecular communication for nanomachines using intercellular calcium
  signaling,'' in \emph{Proc. IEEE NANO}, vol.~2, Aug. 2005, pp. 478--481.

\bibitem{rospars2000perireceptor}
J.-P. Rospars, V.~K{\v{r}}ivan, and P.~L{\'a}nsk{\`y}, ``Perireceptor and
  receptor events in olfaction. comparison of concentration and flux detectors:
  a modeling study,'' \emph{Chem. Senses}, vol.~25, no.~3, pp. 293--311, Jun.
  2000.

\bibitem{bush2010nanoscale}
S.~F. Bush, \emph{Nanoscale Communication Networks}.\hskip 1em plus 0.5em minus
  0.4em\relax Artech House, 2010.

\bibitem{Yilmaz2014effect}
H.~B. Yilmaz, N.-R. Kim, and C.-B. Chae, ``Effect of {ISI} mitigation on
  modulation techniques in molecular communication via diffusion,'' in
  \emph{Proc. ACM NANOCOM}, May 2014, pp. 3:1--3:9.

\bibitem{yilmaz2014three844}
H.~B. Yilmaz, A.~C. Heren, T.~Tugcu, and C.-B. Chae, ``Three-dimensional
  channel characteristics for molecular communications with an absorbing
  receiver,'' \emph{IEEE Commun. Letters}, vol.~18, no.~6, pp. 929--932, Jun.
  2014.

\bibitem{yilmaz2014simulation}
H.~B. Yilmaz and C.-B. Chae, ``Simulation study of molecular communication
  systems with an absorbing receiver: Modulation and {ISI} mitigation
  techniques,'' \emph{Simulat. Modell. Pract. Theory}, vol.~49, pp. 136--150,
  Dec. 2014.

\bibitem{chou2013extended}
C.~T. Chou, ``Extended master equation models for molecular communication
  networks,'' \emph{IEEE Trans. Nanobiosci.}, vol.~12, no.~2, pp. 79--92, Jun.
  2013.

\bibitem{FEDER1980144}
J.~Feder and I.~Giaever, ``Adsorption of ferritin,'' \emph{Journal of Colloid
  and Interface Science}, vol.~78, no.~1, pp. 144 -- 154, Nov. 1980.

\bibitem{ramsden1993concentration}
J.~Ramsden, ``Concentration scaling of protein deposition kinetics,''
  \emph{Physical review letters}, vol.~71, no.~2, p. 295, Jul. 1993.

\bibitem{fang2005kinetics}
F.~Fang, J.~Satulovsky, and I.~Szleifer, ``Kinetics of protein adsorption and
  desorption on surfaces with grafted polymers,'' \emph{Biophysical Journal},
  vol.~89, no.~3, pp. 1516--1533, Jul. 2005.

\bibitem{andrews2009accurate}
S.~S. Andrews, ``Accurate particle-based simulation of adsorption, desorption
  and partial transmission,'' \emph{Physical Biology}, vol.~6, no.~4, p.
  046015, Nov. 2009.

\bibitem{andrews2010detailed}
S.~S. Andrews, N.~J. Addy, R.~Brent, and A.~P. Arkin, ``Detailed simulations of
  cell biology with smoldyn 2.1,'' \emph{PLoS Comput Biol}, vol.~6, no.~3, p.
  e1000705, Mar. 2010.

\bibitem{gul2010nanons}
E.~Gul, B.~Atakan, and O.~B. Akan, ``Nanons: A nanoscale network simulator
  framework for molecular communications,'' \emph{Nano Commun. Net.}, vol.~1,
  no.~2, pp. 138--156, Jun. 2010.

\bibitem{scheider1972two}
W.~Scheider, ``Two-body diffusion problem and applications to reaction
  kinetics,'' \emph{J. Phys. Chem.}, vol.~76, no.~3, pp. 349--361, Feb. 1972.

\bibitem{schulten2000lectures}
K.~Schulten and I.~Kosztin, ``Lectures in theoretical biophysics,''
  \emph{University of Illinois}, vol. 117, 2000.

\bibitem{philip2008biological587}
P.~Nelson, \emph{Biological Physics: Energy, Information, Life}, updated
  1st~ed.\hskip 1em plus 0.5em minus 0.4em\relax W. H. Freeman and Company,
  2008.

\bibitem{Tom2007Multistage}
C.~Tom and M.~R. D'Orsogna, ``Multistage adsorption of diffusing macromolecules
  and viruses,'' \emph{Journal of Chemical Physics}, vol. 127, no.~10, pp.
  2013--2018, 2007.

\bibitem{erban2007reactive}
R.~Erban and S.~J. Chapman, ``Reactive boundary conditions for stochastic
  simulations of reaction--diffusion processes,'' \emph{Physical Biology},
  vol.~4, no.~1, p.~16, Feb. 2007.

\bibitem{wendel1961}
J.~G. Wendel, ``The non-absolute convergence of {G}il-{P}elaez' inversion
  integral,'' \emph{Ann. Math. Stat.}, vol.~32, no.~1, pp. 338--339, Mar. 1961.

\end{thebibliography}

\balance
\end{document}